\documentclass[a4paper,12pt]{article}
\usepackage{amssymb,amsmath,amsthm}
\usepackage{mathrsfs}
\usepackage{fullpage}
\usepackage{hyperref}
\usepackage{eucal}
\usepackage{color}
\usepackage{stackrel}
\usepackage{graphicx}
\newcommand{\ie}{\emph{i.e.}}
\newcommand{\eg}{\emph{e.g.}}
\newcommand{\cf}{\emph{cf.}}
\newcommand{\Real}{\mathbb{R}}

\newcommand{\Nat}{\mathbb{N}}

\newcommand{\supp}{\mathop{\mathrm{supp}}\nolimits}
\newcommand{\Dom}{\mathop{\mathsf{dom}}\nolimits}
\newcommand{\dist}{\mathop{\mathrm{dist}}\nolimits}

\newcommand{\eps}{\varepsilon}
\newcommand{\sii}{L^2}

\newcommand{\der}{\mathrm{d}}
\def\OMIT#1{}
\newtheorem{Theorem}{Theorem}
\newtheorem{Lemma}{Lemma}
\newtheorem{Proposition}{Proposition}

\theoremstyle{definition}
\newtheorem{Remark}{Remark}

\usepackage[normalem]{ulem}
\definecolor{DarkGreen}{rgb}{0,0.5,0.1} 

\newcommand\soutD{\bgroup\markoverwith
{\textcolor{DarkGreen}{\rule[.5ex]{2pt}{1pt}}}\ULon}
\newcommand\soutP{\bgroup\markoverwith
{\textcolor{blue}{\rule[.5ex]{2pt}{1pt}}}\ULon}
\newcommand{\Hm}[1]{\leavevmode{\marginpar{\tiny%
$\hbox to 0mm{\hspace*{-0.5mm}$\leftarrow$\hss}%
\vcenter{\vrule depth 0.1mm height 0.1mm width \the\marginparwidth}%
\hbox to
0mm{\hss$\rightarrow$\hspace*{-0.5mm}}$\\\relax\raggedright #1}}}

\begin{document}
%
\title{\textbf{\Large
Spectral analysis of sheared nanoribbons
}}
\author{Philippe Briet,$^{\!a}$ \ Hamza Abdou Soimadou$^{a}$ \ 
and \ David Krej\v{c}i\v{r}{\'\i}k$^{b}$}
\date{\small 
\emph{
\begin{quote}
\begin{enumerate}
\item[$a)$]
Aix-Marseille Univ, Universit\'e de Toulon, CNRS, CPT, Marseille, France;
briet@univ-tln.fr, abdou-soimadou-hamza@etud.univ-tln.fr
\item[$b)$]
Department of Mathematics, Faculty of Nuclear Sciences and 
Physical Engineering, Czech Technical University in Prague, 
Trojanova 13, 12000 Prague~2, Czechia;
david.krejcirik@fjfi.cvut.cz
\end{enumerate}
\end{quote}
}
\smallskip
26 July 2018}
\maketitle
 
%
\begin{abstract}
\noindent
We investigate the spectrum of the Dirichlet Laplacian in a unbounded strip
subject to a new deformation of ``shearing'':
the strip is built by translating a segment oriented 
in a constant direction along an unbounded curve in the plane.
We locate the essential spectrum under the hypothesis that
the projection of the tangent vector of the curve 
to the direction of the segment admits a (possibly unbounded) limit at infinity  
and state sufficient conditions which guarantee the existence 
of discrete eigenvalues.
We justify the optimality of these conditions by establishing
a spectral stability in opposite regimes.
In particular, Hardy-type inequalities are derived 
in the regime of repulsive shearing.
\bigskip
\begin{itemize}
\item[\textbf{Keywords:}]
sheared strips, quantum waveguides, Hardy inequality, Dirichlet Laplacian.
\item[\textbf{MSC (2010):}]
Primary: 35R45; 81Q10; Secondary: 35J10, 58J50, 78A50.
\end{itemize}
\end{abstract}
%

\section{Introduction} 
%
With advances in nanofabrication and measurement science,
waveguide-shaped nanostructures have reached the point at which
the electron transport can be strongly affected by quantum effects.
Among the most influential theoretical results,
let us quote the existence of quantum bound states due to \emph{bending}
in curved strips, 
firstly observed by Exner and \v{S}eba~\cite{ES} in 1989.
The pioneering paper has been followed by a huge number 
of works demonstrating the robustness of the effect 
in various geometric settings including higher dimensions,
and the research field is still active these days
(see~\cite{KTdA2} for a recent paper
with a brief overview in the introduction).

In 2008 Ekholm, Kova\v{r}\'ik and one of the present authors~\cite{EKK}
observed that the geometric deformation of \emph{twisting} 
has a quite opposite effect on the energy spectrum 
of an electron confined to three-dimensional tubes,
for it creates an effectively repulsive interaction 
(see~\cite{K6-with-erratum} for an overview of 
the two reciprocal effects). 
More specifically, twisting the tube locally gives rise 
to Hardy-type inequalities and a stability of quantum transport,
the effect becomes stronger in globally twisted tubes~\cite{BHK}
and in extreme situations it may even annihilate 
the essential spectrum completely~\cite{K11}
(see also \cite{Barseghyan-Khrabustovskyi_2018}).
The repulsive effect remains effective even under modification 
of the boundary conditions~\cite{Briet-Hammedi_2015}.

The objective of this paper is to introduce a new, two-dimensional model 
exhibiting a previously unobserved geometric effect of \emph{shearing}.
Mathematically, it is reminiscent of the effect of twisting 
in the three-dimensional tubes 
(in some aspects it also recalls 
the geometric setting of curved wedges studied in~\cite{K12}), 
but the lower dimensional simplicity
enables us to get an insight into analogous problems left open in~\cite{BHK}
and actually provide a complete spectral picture now. 
The richness of the toy model is reflected in covering 
very distinct regimes, ranging from purely essential 
to purely discrete spectra or a combination of both.
We believe that the present study will stimulate 
further interest in sheared nanostructures. 

The model that we consider in this paper is characterised
by a positive number~$d$ (the transverse width of the waveguide) 
and a differentiable function $f:\Real\to\Real$ 
(the boundary profile of the waveguide).
The waveguide~$\Omega$ is introduced as the set of points in~$\Real^2$
delimited by the curve $x \mapsto (x,f(x))$ 
and its vertical translation $x \mapsto (x,f(x)+d)$, 
namely (see Figure~\ref{figA}), 
\begin{equation}\label{Omega}
  \Omega := \left\{ 
  (x,y) \in \Real^2 : f(x) < y < f(x)+d
  \right\}
  .
\end{equation}

We stress that the geometry of~$\Omega$ 
differs from the curved strips intensively studied
in the literature for the last thirty years
(see~\cite{KKriz} for a review).
In the latter case, the segment $(0,d)$ is translated along 
the curve $x \mapsto (x,f(x))$ with respect to its \emph{normal} vector field
(so that the waveguide is delimited by two \emph{parallel} curves),
while in the present model the translation is with respect to 
the \emph{constant} basis vector in the $y$-direction.

\begin{figure}[h]
\begin{center}
\includegraphics[width=0.9\textwidth]{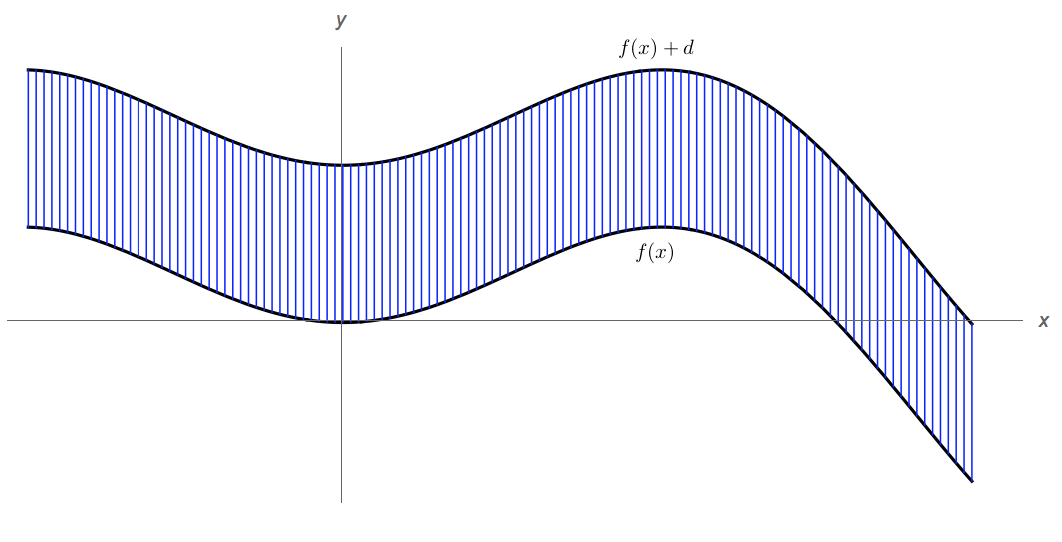}
\caption{The geometry of a sheared waveguide.}\label{figA}
\end{center}
\end{figure}

Clearly, it is rather the derivative~$f'$ that determines the shear	
deformation of the straight waveguide $\Omega_0 := \Real \times (0,d)$.
Our standing assumption is that the derivative~$f'$ 
is locally bounded and that it admits a (possibly infinite) 
limit at infinity:
\begin{equation}\label{Ass}
  \lim_{|x| \to \infty} f'(x) =: \beta \in \Real \cup \{\pm \infty\} 
  .
\end{equation}
If~$\beta$ is finite,
we often denote the deficit 
\begin{equation}\label{eps}
  \eps(x):=f'(x)-\beta
  \,,
\end{equation}
but notice that~$\eps(x)$ is not necessarily small 
for finite~$x$.
	
We consider an effectively free electron constrained to
the nanostructure~$\Omega$ by hard-wall boundaries.
Disregarding physical constants, the quantum Hamiltonian of 
the system can be identified with the Dirichlet Laplacian 
$-\Delta_D^\Omega$ in $\sii(\Omega)$.
The spectrum of the straight waveguide~$\Omega_0$ 
(which can be identified with~$f'=0$ in our model)
is well known; by separation of variables,
one easily conclude that
$\sigma(-\Delta_D^{\Omega_0}) = [(\frac{\pi}{d})^2,\infty)$
and the spectrum is purely absolutely continuous.
The main spectral properties of $-\Delta_D^\Omega$
under the shear deformation obtained in this paper
are summarised in the following theorems.

First, we locate the possible range of energies of propagating states.
\begin{Theorem}[Essential spectrum]\label{Thm.ess}
Let $f' \in L_\mathrm{loc}^\infty(\Real)$ satisfy~\eqref{Ass}.
Then
\begin{equation}\label{ess}
  \sigma_\mathrm{ess}(-\Delta_D^\Omega) = [E_1(\beta),\infty)
  \,, \qquad \mbox{where} \qquad
  E_1(\beta) := (1+\beta^2) \left(\frac{\pi}{d}\right)^2 \,.
\end{equation}
\end{Theorem}

If $\beta = \pm\infty$ (see Figure~\ref{figB}), 
the result means $\sigma_\mathrm{ess}(-\Delta_D^\Omega) = \varnothing$,
so the spectrum is purely discrete 
and there are no propagating states. 
In this case, the distance of points $\mathrm{x} \in \Omega$
to the boundary~$\partial\Omega$ tends to zero as $|\mathrm{x}| \to \infty$,
so~$\Omega$ is a quasi-bounded domain (\cf~Remark~\ref{Rem.q-bounded}).
This phenomenon is reminiscent of three-dimensional waveguides
with asymptotically diverging twisting 
\cite{K11,Barseghyan-Khrabustovskyi_2018}.

\begin{figure}[h]
\begin{center}
\includegraphics[width=0.9\textwidth]{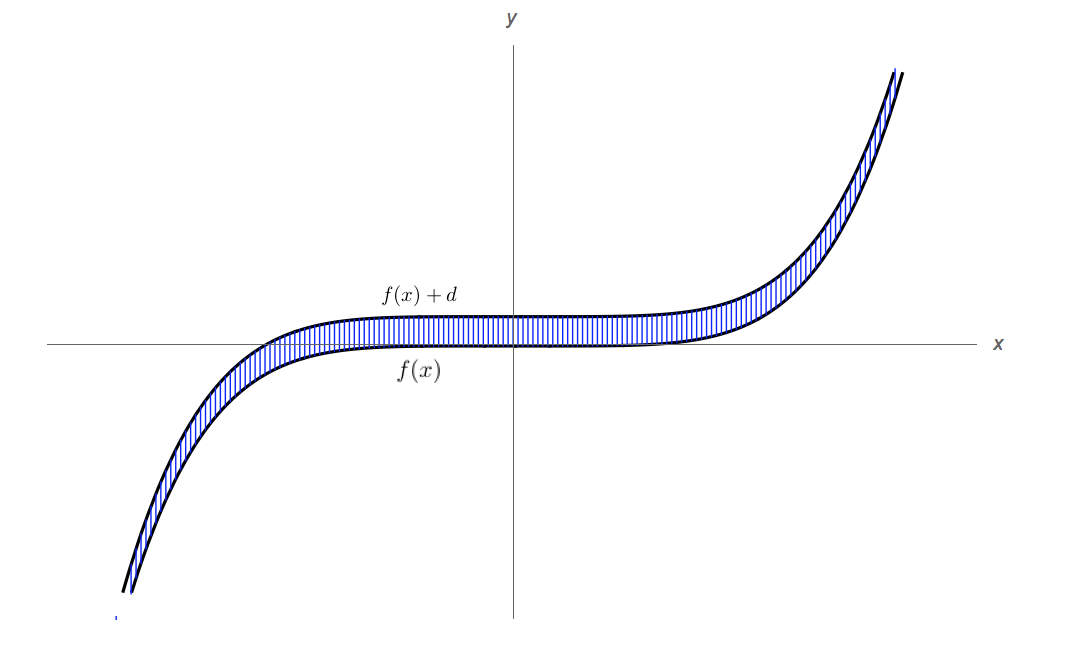}
\caption{A sheared waveguide with $\beta=+\infty$.}\label{figB}
\end{center}
\end{figure}

Our next concern is about the possible existence or absence
of discrete eigenvalues below $E_1(\beta)$ if~$\beta$ is finite.
The following theorem together with Theorem~\ref{Thm.ess}
provides a sufficient condition for the existence.

\begin{Theorem}[Attractive shearing]\label{Thm.disc}
Let $f' \in L_\mathrm{loc}^\infty(\Real)$ 
be such that $f'^2 - \beta^2 \in L^1(\Real)$.
Then either of the following conditions
\begin{enumerate}
\item[\emph{(i)}]
$
\displaystyle
  \int_{\Real} (f'^2(x) - \beta^2) \, \der x < 0
  \,;
$
\item[\emph{(ii)}]
$
\displaystyle
  \int_{\Real} (f'^2(x) - \beta^2) \, \der x = 0
$,
$f'$~is not constant 
and $f'' \in L_\mathrm{loc}^1(\Real) \,;$
\end{enumerate}
implies 
\begin{equation}\label{less}
  \inf\sigma(-\Delta_D^\Omega) < E_1(\beta)
  \,.  
\end{equation}
\end{Theorem}

The theorem indeed implies the existence of discrete spectra
under the additional hypothesis~\eqref{Ass},
because the inequality~\eqref{less}
and the fact that, due to Theorem~\ref{Thm.ess},
the essential spectrum of $-\Delta_D^\Omega$ starts by $E_1(\beta)$
ensure that the lowest point in the spectrum of $-\Delta_D^\Omega$
is an isolated eigenvalue of finite multiplicity.
Theorem~\ref{Thm.disc} is an analogy of~\cite{EKov_2005},
where it was shown that a local ``slow down'' of twisting 
in periodically twisted tubes generates discrete eigenvalues
(see also \cite{Briet-Kovarik-Raikov-Soccorsi_2009}).
It is not clear from our proof whether the extra regularity
assumption in the condition~(ii) is necessary. 

Recalling~\eqref{eps},
notice that $f'^2-\beta^2 = \eps^2 + 2 \beta \eps$,
so a necessary condition to satisfy either of 
the conditions~(i) or~(ii) of Theorem~\ref{Thm.disc}
is that the function $\beta \eps$ is not non-negative.  
This motivates us to use the terminology that 
the shear is \emph{attractive} (respectively, \emph{repulsive})
if~$\beta\eps$ is non-positive (respectively, non-negative)
and~$\eps$ is non-trivial.
In analogy with the conjectures stated in~\cite{BHK},
it is expected that the discrete spectrum of $-\Delta_D^\Omega$
is empty provided that the shear is either repulsive 
or possibly attractive but $\eps^2 \gg -\beta\eps$. 
The latter is confirmed in the following special setting
(see Figure~\ref{figC}).

\begin{Theorem}[Strong shearing]\label{Thm.schema}
Let $f'(x) = \beta + \alpha \eps(x)$, where 
$\alpha, \beta \in \Real$ and $\eps:\Real \to \Real$ 
is a function such that
$\supp \eps \subset [0,1]$ and $c_1 \leq \eps(x) \leq c_2$
for all $x\in[0,1]$ with some positive constants $c_1,c_2$.
Then there exists a positive number~$\alpha_0$ such that,
for all $|\alpha| \geq \alpha_0$,
\begin{equation}\label{stability}
  \sigma(-\Delta_D^\Omega) = [E_1(\beta),\infty)
  \,.  
\end{equation}
\end{Theorem}

Finally, we establish Hardy-type inequalities 
in the case of repulsive shear.

\begin{Theorem}[Repulsive shearing]\label{Thm.Hardy}
Let $f' \in L_\mathrm{loc}^\infty(\Real)$.
If $\beta\eps \geq 0$ and~$\eps$ is non-trivial, 
then there exists a positive constant~$c$ such that the inequality
\begin{equation}\label{Hardy}
  -\Delta_D^\Omega - E_1(\beta) \geq \frac{c}{1+x^2}
\end{equation}
holds in the sense of quadratic forms in $\sii(\Omega)$.
\end{Theorem}

Consequently, if the shear satisfies~\eqref{Ass}
in addition to the repulsiveness,
then~\eqref{Hardy} and Theorem~\ref{Thm.ess} imply 
that the stability result~\eqref{stability} holds,
so in particular there are no discrete eigenvalues.
Moreover, the spectrum is stable against small perturbations
(the smallness being quantified by the Hardy weight
on the right-hand side of~\eqref{Hardy}).
From Theorem~\ref{Thm.Hardy} it is also clear that~$\alpha$
does not need to be large if $\alpha\beta \geq 0$
in the setting of Theorem~\ref{Thm.schema}.
 
The organisation of the paper is as follows.
In Section~\ref{Sec.pre} we introduce natural curvilinear coordinates
to parameterise~$\Omega$ and express the Dirichlet Laplacian in them.
Hardy-type inequalities (and in particular Theorem~\ref{Thm.Hardy}) 
are established in Section~\ref{Sec.Hardy}. 
The essential spectrum (Theorem~\ref{Thm.ess}) 
is l in Section~\ref{Sec.ess}.
The sufficient conditions of Theorem~\ref{Thm.disc}
(which in particular imply the existence of discrete eigenvalues
in the regime~\eqref{ess})
are established in Section~\ref{Sec.disc}.
Finally, Theorem~\ref{Thm.schema} is proved in Section~\ref{Sec.schema}.

\section{Preliminaries}\label{Sec.pre}
%
Our strategy to deal with the Laplacian in the sheared geometry~$\Omega$
is to express it in natural curvilinear coordinates. 
It employs the identification 
$
  \Omega = \mathscr{L}(\Omega_0)
$,
where $\mathscr{L}:\Real^2\to\Real^2$ is the shear mapping defined by
\begin{equation}\label{tube}
  \mathscr{L}(s,t) := \big(s,f(s)+t\big)
  .
\end{equation}
In this paper we consistently use the notations~$s$ and~$t$
to denote the ``longitudinal'' and ``transversal'' 
variables in the straight waveguide~$\Omega_0$, respectively. 
The corresponding metric reads
\begin{equation}
  G := \nabla\mathscr{L} \cdot (\nabla\mathscr{L})^T = 
  \begin{pmatrix}
    1+f'^2 & f' 
    \\
    f' & 1
  \end{pmatrix}
  , \qquad
  |G| := \det(G) = 1
  ,
\end{equation}
where the dot denotes the scalar product in~$\Real^2$.
Recall our standing assumption 	that $f' \in L_\mathrm{loc}^\infty(\Real)$.
It follows that $\mathscr{L}:\Omega_0 \to \Omega$ is a local diffeomorphism.
In fact, it is a global diffeomorphism, because~$\mathscr{L}$ is injective.
In this way, we have identified~$\Omega$ 
with the Riemannian manifold $(\Omega_0,G)$.

The Dirichlet Laplacian $-\Delta_D^\Omega$ is introduced standardly
as the self-adjoint operator in the Hilbert space $\sii(\Omega)$
associated with the quadratic form $Q_D^\Omega[u] := \int_\Omega |\nabla u|^2$,
$\Dom Q_D^\Omega := H_0^1(\Omega)$. 
Employing the diffeomorphism $\mathscr{L}:\Omega_0 \to \Omega$,
we have the unitary transform
$$
  \mathcal{U} : \sii(\Omega) \to \sii(\Omega_0) :
  \{u \mapsto u \circ \mathscr{L} \}
  \,.
$$
Consequently, $-\Delta_D^\Omega$ is unitarily equivalent 
(and therefore isospectral)
to the operator $H := \mathcal{U}(-\Delta_D^\Omega)\mathcal{U}^{-1}$
in the Hilbert space $\sii(\Omega_0)$.
By definition, the operator~$H$ is associated with the quadratic form
$$  
  h[\psi] := Q_D^\Omega[\mathcal{U}^{-1}\psi] 
  \,, \qquad
  \Dom h := \mathcal{U} \Dom Q_D^\Omega
  \,.	
$$
Occasionally, we shall emphasise the dependence of~$h$ and~$H$ on~$f'$
by writing~$f'$ as the subscript, 
\ie\ $h=h_{f'}$ and $H=H_{f'}$.

Given $u \in C_0^\infty(\Omega)$, a core of $Q_D^\Omega$,
it follows that $\psi := u \circ \mathscr{L}$
is also compactly supported and $\psi \in H_0^1(\Omega_0)$.
Then it is straightforward to verify that
\begin{equation}\label{form}
  h[\psi] = \|\partial_s\psi-f'\partial_t\psi\|^2 + \|\partial_t\psi\|^2
  \,,
\end{equation}
where $\|\cdot\|$ denotes the norm of $\sii(\Omega_0)$
and, with an abuse of notation, we denote by the same symbol~$f'$
the function $f' \otimes 1$ on $\Real\times(0,d)$.
We have the elementary bounds
\begin{multline}\label{elementary}
  \delta \|\partial_s\psi\|^2 
  + \left(1-\frac{\delta}{1-\delta}\|f'\|_{L^\infty(\supp\psi)}^2\right) 
  \|\partial_t\psi\|^2
  \\
  \leq h[\psi] \leq
  2 \|\partial_s\psi\|^2 
  + \left(1 + 2 \|f'\|_{L^\infty(\supp\psi)}^2\right) \|\partial_t\psi\|^2
\end{multline}
valid for every $\delta \in (0,1)$.	
The domain of~$h$ coincides with the closure of the set 
of the functions~$\psi$ specified above	
with respect to the graph topology of~$h$.
By a standard mollification argument using~\eqref{elementary}, 	
it follows that~$C_0^\infty(\Omega_0)$ is a core of~$h$. 	
Moreover, it follows from~\eqref{elementary} that
if~$f'$ is bounded, then the graph topology of~$h$
is equivalent to the topology of $H^1(\Omega_0)$,
and therefore $\Dom h = H_0^1(\Omega_0)$ in this case.
In general (and in particular if~\eqref{Ass} holds with $\beta = \pm\infty$),
however, the domain of~$h$ is not necessarily equal to $H_0^1(\Omega_0)$.  
Let us summarise the preceding observations into the following proposition.

\begin{Proposition}\label{Prop.core}
Let $f' \in L_\mathrm{loc}^\infty(\Real)$.
Then $C_0^\infty(\Omega_0)$ is a core of~$h$.
If moreover $f' \in L^\infty(\Real)$, then $\Dom h = H_0^1(\Omega_0)$. 
\end{Proposition}

In a distributional sense, 
we have $H = - (\partial_s-f'\partial_t)^2 - \partial_t^2$,
but we shall use neither this formula nor any information
on the operator domain of~$H$.
Our spectral analysis of~$H$ will be exclusively based on
the associated quadratic form~\eqref{form}.
Notice that the structure of~$H$ is similar to 
twisted three-dimensional tubes~\cite{EKK}
as well as curved two-dimensional wedges~\cite{K12}.	

\section{Hardy-type inequalities}\label{Sec.Hardy}
%
In this section we prove Theorem~\ref{Thm.Hardy}.
The strategy is to first establish a ``local'' Hardy-type inequality,
for which the Hardy weight might not be everywhere positive,
and then ``smeared it out'' to the ``global'' inequality~\eqref{Hardy}
with help of a variant of the classical one-dimensional Hardy inequality.
Throughout this section, we assume that $\beta \in \Real$
and $\eps \in L_\mathrm{loc}^\infty(\Real)$.	 

Let 
$$
  E_1 := \left(\frac{\pi}{d}\right)^2
  \qquad \mbox{and} \qquad
  \chi_1(t) := \sqrt{\frac{2}{d}} \, \sin(E_1 t)
$$
denote respectively the lowest eigenvalue and the corresponding
normalised eigenfunction of the Dirichlet Laplacian in $\sii((0,d))$.  
By the variational definition of~$E_1$,
we have the Poincar\'e-type inequality
\begin{equation}\label{Poincare}
  \forall \chi \in H_0^1((0,d)) \,, \qquad
  \int_0^d |\chi'(t)|^2 \, \der t \geq E_1 \int_0^d |\chi(t)|^2 \, \der t
  \,.
\end{equation}
Notice that $E_1 = E_1(0)$, where the latter is introduced in~\eqref{ess}.
Neglecting the first term on the right-hand side of~\eqref{form}
and using~\eqref{Poincare} for the second term
together with Fubini's theorem, we immediately get the lower bound
\begin{equation}\label{robust}
  H \geq E_1
  \,,
\end{equation}
which is independent of~$f$.	

The main ingredient in our approach is the following valuable lemma 
which reveals a finer structure of the form~$h$.
With an abuse of notation, we use the same symbol~$\chi_1$
for the function $1 \otimes \chi_1$ on $\Real \times (0,d)$,
and similarly for its derivative~$\chi_1'$,
while~$\eps$ also denotes the function 
$\eps \otimes 1$ on $\Real \times (0,d)$. 

\begin{Lemma}[Ground-state decomposition]\label{Lem.fine}
For every $\psi \in C_0^\infty(\Omega_0)$, we have
\begin{equation}\label{fine}
\begin{aligned}
  h[\psi] - E_1(\beta) \|\psi\|^2
  = \ & \|\partial_s\psi
  - \eps \partial_t \psi - \beta \chi_1\partial_t(\chi_1^{-1}\psi) \|^2 
  + \|\chi_1\partial_t(\chi_1^{-1}\psi)\|^2
  \\
  & + \int_{\Omega_0} \beta\,\eps(s)\left[
  E_1 + \left(\frac{\chi_1'(t)}{\chi_1(t)}\right)^2
  \right]
  |\psi(s,t)|^2 \, \der s \, \der t
  .
\end{aligned}
\end{equation}
\end{Lemma}
\begin{proof}
Let $\phi\in C_0^\infty(\Omega_0)$ be defined 	
by the decomposition $\psi(s,t) = \chi_1(t) \phi(s,t)$.
Writing $2\Re(\phi\partial_t\phi) = \partial_t |\phi|^2$,
integrating by parts and using the differential equation
that~$\chi_1$ satisfies, we have 
$$
  \|\partial_t\psi\|^2 = \|\chi_1'\phi\|^2 + \|\chi_1\partial_t\phi\|^2
  + 2\Re(\chi_1'\phi,\chi_1\partial_t\phi)
  = \|\chi_1\partial_t\phi\|^2 + E_1 \|\psi\|^2
  \,,
$$
where $(\cdot,\cdot)$ denotes the inner product of $\sii(\Omega_0)$.
Similarly,
$$
\begin{aligned}
  \|\partial_s\psi-f'\partial_t\psi\|^2
  =\ & \|\partial_s\psi-\eps\partial_t\psi - \beta \chi_1\partial_t\phi\|^2 
  + \beta^2 \|\chi_1'\phi\|^2
  \\
  &- 2\beta \, \Re(\chi_1\partial_s\phi
  -\eps\chi_1\partial_t\phi - \eps \chi_1'\phi - \beta \chi_1\partial_t\phi,
  \chi_1'\phi)
  \\
  =\ & \|\partial_s\psi-\eps\partial_t\psi - \beta \chi_1\partial_t\phi\|^2 
  \\
  & + \beta E_1 (\eps\chi_1\phi,\chi_1\phi)
  +\beta (\eps\chi_1'\phi,\chi_1'\phi)
  + \beta^2 E_1 (\chi_1\phi,\chi_1\phi)
  \\
  =\ & \|\partial_s\psi-\eps\partial_t\psi - \beta \chi_1\partial_t\phi\|^2 
  \\
  & + \beta E_1 (\eps\psi,\psi)
  + \beta (\eps\chi_1'\phi,\chi_1'\phi)
  + \beta^2 E_1 \|\psi\|^2 
  \,,
\end{aligned}
$$
where the second equality follows by integrations by parts 
in the cross term.
Summing up the results of the two computations
and subtracting $E_1(\beta) \|\psi\|^2$, 
we arrive at the desired identity.
\end{proof}

From now on, let us assume that the shear is repulsive, 
\ie\ $\beta\eps \geq 0$ and $\eps \not= 0$.
If $\beta\not=0$, then Lemma~\ref{Lem.fine} immediately
gives the following local Hardy-type inequality
\begin{equation}\label{Hardy.local}
  H - E_1(\beta) \geq \beta\eps 
  \left[
  E_1 + \left(\frac{\chi_1'}{\chi_1}\right)^2
  \right]
  .
\end{equation}
Recall that~$E_1(\beta)$ corresponds to the threshold
of the essential spectrum if~\eqref{Ass} holds (\cf~Theorem~\ref{Thm.ess}),
so~\eqref{Hardy.local} in particular ensures that 
there is no (discrete) spectrum below~$E_1(\beta)$
(even if $\beta=0$). 
Note that the Hardy weight on the right-hand side of~\eqref{Hardy.local}
diverges on $\partial\Omega_0$.
The terminology ``local'' comes from the fact that 
the Hardy weight might not be everywhere positive 
(\eg~if $\beta\eps$ is compactly supported).	

In order to obtain a non-trivial non-negative lower bound
including the case $\beta=0$, 
we have to exploit the positive terms
that we have neglected when coming from~\eqref{fine} to~\eqref{Hardy.local}.	
To do so, let $I \subset \Real$ be any bounded open interval 
and set $\Omega_0^I := I\times(0,d)$.
In the Hilbert space $\sii(\Omega_0^I)$ let us consider
the quadratic form
$$
\begin{aligned}
  q_I[\psi] 
  &:= 
  \|\partial_s\psi
  - \eps \partial_t \psi 
  - \beta \chi_1\partial_t(\chi_1^{-1}\psi) \|_{\sii(\Omega_0^I)}^2 
  + \|\chi_1\partial_t(\chi_1^{-1}\psi)\|_{\sii(\Omega_0^I)}^2
  \,, \\
  \Dom q_I
  &:= \{\psi \upharpoonright \Omega_0^I:
  \psi \in H_0^1(\Omega_0)\}
  \,.
\end{aligned}
$$
The form corresponds to the sum of the two first terms 
on the right-hand side of~\eqref{fine}
with the integrals restricted to~$\Omega_0^I$.

\begin{Lemma}\label{Lem.compact}
The form~$q_I$ is closed. 
\end{Lemma}
\begin{proof}
First of all, notice that 
the form~$q_I$ is obviously closed if $\beta=0$
(one can proceed as in~\eqref{elementary}). 
In general, given $\psi \in \Dom q_I$,
we employ the estimates
\begin{multline*}
  \delta \|\partial_s\psi - \eps \partial_t \psi\|_{\sii(\Omega_0^I)}^2
  + \left(1-\frac{\delta}{1-\delta}\,	\beta^2\right) 
  \|\chi_1\partial_t(\chi_1^{-1}\psi)\|_{\sii(\Omega_0^I)}^2
  \\
  \leq q_I[\psi] \leq
  2 \|\partial_s\psi - \eps \partial_t \psi\|_{\sii(\Omega_0^I)}^2
  + (1+2\beta^2)\|\chi_1\partial_t(\chi_1^{-1}\psi)\|_{\sii(\Omega_0^I)}^2
\end{multline*}
valid for every $\delta \in (0,1)$  
and the identity 
(easily checked by employing the density of 
$C_0^\infty(\Omega_0)$ in $H_0^1(\Omega_0)$)
\begin{equation}\label{minusE1}
  \|\chi_1\partial_t(\chi_1^{-1}\psi)\|_{\sii(\Omega_0^I)}^2
  = \|\partial_t\psi\|_{\sii(\Omega_0^I)}^2 - E_1 \|\psi\|_{\sii(\Omega_0^I)}^2
  \,.
\end{equation}
Instead of the latter, we could also use the fact that
$\chi_1^{-1}\psi \in \sii(\Omega_0)$ whenever $\psi \in H_0^1(\Omega_0)$,
\cf~\cite[Thm.~1.5.6]{Davies_1989}.
Using in addition the aforementioned fact that~$q_I$ is closed if $\beta=0$,
we conclude that~$q_I$ is closed for any value $\beta \in \Real$.
\end{proof}

Consequently, the form~$q_I$
is associated to a self-adjoint operator~$Q_I$ with compact resolvent.
Let us denote by~$\lambda_I$ its lowest eigenvalue, \ie,
\begin{equation}\label{lambda}
  \lambda_I := \inf_{\stackrel[\psi\not=0]{}{\psi \in \Dom q_I}} 
  \frac{q_I[\psi]}
  {\, \|\psi\|_{\sii(\Omega_0^I)}^2} 
  \,.
\end{equation}
The eigenvalue~$\lambda_I$ is obviously non-negative.
The following lemma shows that~$\lambda_I$ is actually positive
whenever~$\eps$ is non-trivial on~$I$. 

\begin{Lemma}\label{Lem.lambda}
$\lambda_I=0$ if, and only if, $\eps=0$ on~$I$. 
\end{Lemma}
\begin{proof}
If $\eps \upharpoonright I = 0$, 
then the function $\psi(s,t):=\chi_1(t)$
minimises~\eqref{lambda} with $\lambda_I=0$.
Conversely, let us assume that $\lambda_I=0$
and denote by~$\psi_1$ the corresponding eigenfunction of~$Q_I$,
\ie~the minimiser of~\eqref{lambda}.  
It follows from~\eqref{lambda} together with~\eqref{minusE1} that 	
\begin{equation}\label{two}
\begin{aligned}
  \|\partial_s\psi_1
  - \eps \partial_t \psi_1 
  - \beta \chi_1\partial_t(\chi_1^{-1}\psi_1) \|_{\sii(\Omega_0^I)}^2 &= 0
  \,,
  \\
  \|\partial_t\psi_1\|_{\sii(\Omega_0^I)}^2
  -E_1\|\psi_1\|_{\sii(\Omega_0^I)}^2 &= 0
  \,.
\end{aligned}
\end{equation}
Writing $\psi_1(s,t) = \varphi(s)\chi_1(t) + \phi(s,t)$,
where
$$
  \int_0^d \chi_1(t)\phi(s,t) \, \der t = 0
$$
for almost every $s \in I$, we have 
$$
  \|\partial_t\psi_1\|_{\sii(\Omega_0^I)}^2
  -E_1\|\psi_1\|_{\sii(\Omega_0^I)}^2  
  = \|\partial_t\phi\|_{\sii(\Omega_0^I)}^2
  -E_1\|\phi\|_{\sii(\Omega_0^I)}^2  
  \geq (E_2-E_1) \|\phi\|_{\sii(\Omega_0^I)}^2
  \,,
$$
where $E_2=4E_1$ 
is the second eigenvalue of the Dirichlet Laplacian in~$\sii((0,d))$.
Since the gap $E_2-E_1$ is positive, 
it follows from the second identity in~\eqref{two} that $\phi=0$.
Plugging now the separated-eigenfunction Ansatz
$\psi_1(s,t) = \varphi(s)\chi_1(t)$
into the first identity in~\eqref{two}
and integrating by parts with respect to~$t$ in the cross term, 
we obtain
$$
  0 = \|\partial_s\psi_1-\eps\partial_t\psi_1\|_{\sii(\Omega_0^I)}^2 
  = \|\varphi'\chi_1\|_{\sii(\Omega_0^I)}^2 
  + \|\eps\varphi\chi_1'\|_{\sii(\Omega_0^I)}^2 
  \,.
$$
Consequently, $\varphi$~is necessarily constant 
and $\|\eps\|_{\sii(I)}=0$.  
\end{proof}

Using~\eqref{lambda}, we can improve~\eqref{Hardy.local} 
to the following local Hardy-type inequality
\begin{equation}\label{Hardy.local.improved}
  H - E_1(\beta) \geq 
  \beta\eps 
  \left[
  E_1 + \left(\frac{\chi_1'}{\chi_1}\right)^2
  \right]
  + \lambda_I \, 1_{I\times(0,d)}
  \,,
\end{equation}
where $1_M$~denotes the characteristic function of a set~$M$.
The inequality is valid with any bounded interval $I \subset \Real$,
but recall that~$\lambda_I$ is positive if, and only if, $I$~is chosen in 
such a way that~$\eps$ is not identically equal to zero on~$I$,
\cf~Lemma~\ref{Lem.lambda}. 

The passage from the local Hardy-type inequality~\eqref{Hardy.local.improved}
to the global Hardy-type inequality of Theorem~\ref{Thm.Hardy}
will be enabled by means
of the following crucial lemma, 
which is essentially a variant of the classical
one-dimensional Hardy inequality.

\begin{Lemma}\label{Lem.1DHardy}
Let $s_0 \in \Real$ and
$\psi \in C_0^\infty(\Omega_0\setminus\{s_0\}\times(0,d))$. 
Then
\begin{equation}\label{1DHardy}
  \|\partial_s\psi
  - \eps \partial_t \psi - \beta \chi_1\partial_t(\chi_1^{-1}\psi)\|^2 
  + \|\chi_1\partial_t(\chi_1^{-1}\psi)\|^2
  \geq \frac{1}{4(1+\beta^2)} \int_{\Omega_0} 
  \frac{|\psi(s,t)|^2}{(s-s_0)^2} \, \der s \, \der t 
  \,.
\end{equation}
\end{Lemma}
\begin{proof} 
Denote $\rho(s,t):=(s-s_0)^{-1}$.
Given any real number~$\alpha$, we write
\begin{align*} 
  \lefteqn{ \|\partial_s\psi
  - \eps \partial_t \psi - \beta \chi_1\partial_t(\chi_1^{-1}\psi)
  - \alpha \rho \psi\|^2}
  \\
  =\ & \|\partial_s\psi
  - \eps \partial_t \psi - \beta \chi_1\partial_t(\chi_1^{-1}\psi)\|^2
  + \alpha^2 \| \rho \psi \|^2
  - 2 \alpha \Re (
  \partial_s\psi
  - \eps \partial_t \psi - \beta \chi_1\partial_t(\chi_1^{-1}\psi),
  \rho\psi
  )
  \\
  =\ & \|\partial_s\psi
  - \eps \partial_t \psi - \beta \chi_1\partial_t(\chi_1^{-1}\psi)\|^2
  + (\alpha^2-\alpha) \| \rho\psi \|^2
  + 2\alpha\beta \Re(\chi_1\partial_t(\chi_1^{-1}\psi),\rho\psi)
  \\
  \leq \ &
  \left\|\partial_s\psi
  - \eps \partial_t \psi - \beta \chi_1\partial_t(\chi_1^{-1}\psi)\right\|^2
  + (\alpha^2-\alpha+\alpha^2\beta^2) \| \rho\psi \|^2
  + \|\chi_1\partial_t(\chi_1^{-1}\psi)\|^2	
  \,,
\end{align*}
where the second equality follows by integrations by parts
and the estimate is due to the Schwarz inequality.
Since the first line is non-negative, 
we arrive at the desired inequality
by choosing the optimal $\alpha := [2(1+\beta^2)]^{-1}$.
\end{proof}

Now we are in a position to prove the ``global'' Hardy inequality 
of Theorem~\ref{Thm.Hardy}.
\begin{proof}[Proof of Theorem~\ref{Thm.Hardy}]
Let $\psi \in C_0^\infty(\Omega_0)$.
First of all, since~$\eps$ is assumed to be non-trivial, 
there exists a real number~$s_0$ and positive~$b$
such that $\eps$ restricted to the interval 
$I := (s_0-b,s_0+b)$ is non-trivial.
By Lemma~\ref{lambda}, it follows that~$\lambda_I$ is positive.
Then the local Hardy-type inequality~\eqref{Hardy.local.improved} implies
the estimate
\begin{equation}\label{in1}
  h[\psi] - E_1(\beta) \|\psi\|^2
  \geq \lambda_I \|\psi\|_{\sii(\Omega_0^I)}^2
  \,.
\end{equation}

Second,
let $\eta \in C^\infty(\Real)$ be such that $0 \leq \eta \leq 1$,
$\eta=0$ in a neighbourhood of~$s_0$ and $\eta=1$ outside~$I$. 
Let us denote by the same symbol~$\eta$ the function 
$\eta \otimes 1$ on $\Real \times (0,d)$,
and similarly for its derivative~$\eta'$. 
Writing $\psi = \eta\psi + (1-\eta)\psi$, 
we have
\begin{equation}\label{writing}
  \int_{\Omega_0} 
  \frac{|\psi(s,t)|^2}{1+(s-s_0)^2} \, \der s \, \der t
  \leq 2 \int_{\Omega_0} 
  \frac{|(\eta\psi)(s,t)|^2}{(s-s_0)^2} \, \der s \, \der t 
  + 2 \int_{\Omega_0} 
  |((1-\eta)\psi)(s,t)|^2\, \der s \, \der t 
  \,.
\end{equation}
Since~$\eta\psi$ satisfies the hypothesis of Lemma~\ref{Lem.1DHardy}, 
we have 
\begin{align*}
  \lefteqn{
  \int_{\Omega_0}
  \frac{|(\eta\psi)(s,t)|^2}{(s-s_0)^2} \, \der s \, \der t 
  }
  \\
  &\leq 4 (1+\beta^2) \left(
  \|\partial_s(\eta\psi)
  - \eta\eps \partial_t \psi - \eta\beta \chi_1\partial_t(\chi_1^{-1}\psi)\|^2 
  + \|\eta\chi_1\partial_t(\chi_1^{-1}\psi)\|^2
  \right)
  \\
  &\leq 8 (1+\beta^2) \|\partial_s\psi
  - \eps \partial_t \psi - \beta \chi_1\partial_t(\chi_1^{-1}\psi)\|^2 
  + 8 (1+\beta^2) \|\eta'\psi\|
  + 4 (1+\beta^2)  \|\chi_1\partial_t(\chi_1^{-1}\psi)\|^2	
  \\
  &\leq 8 (1+\beta^2) (h[\psi]-E_1(\beta)\|\psi\|^2)
  + 8 (1+\beta^2) \|\eta'\|_{\infty} \|\psi\|_{\sii(\Omega_0^I)}^2
\end{align*}
where $\|\eta'\|_{\infty}$ denotes the supremum norm of~$\eta'$.
Here the last inequality employs Lemma~\ref{Lem.fine}
(recall that we assume that the shear is repulsive). 
At the same time,
$$
  \int_{\Omega_0} 
  |((1-\eta)\psi)(s,t)|^2\, \der s \, \der t 
  \leq \|\psi\|_{\sii(\Omega_0^I)}^2
  \,.
$$
From~\eqref{writing} we therefore deduce	
\begin{equation}\label{in2}
  h[\psi] - E_1(\beta) \|\psi\|^2
  \geq 
  \frac{1}{16(1+\beta^2)} \int_{\Omega_0} 
  \frac{|\psi(s,t)|^2}{1+(s-s_0)^2} \, \der s \, \der t
  - \left(
  \|\eta'\|_\infty^2 + \frac{1}{8(1+\beta^2)}
  \right)
  \|\psi\|_{\sii(\Omega_0^I)}^2
  \,.
\end{equation}

Finally, interpolating between~\eqref{in1} and~\eqref{in2},
we obtain 
\begin{align*}
  h[\psi] - E_1(\beta) \|\psi\|^2
  \geq \ & 
   \frac{\delta}{16(1+\beta^2)} \int_{\Omega_0} 
  \frac{|\psi(s,t)|^2}{1+(s-s_0)^2} \, \der s \, \der t
  \\
  & + \left[
  (1-\delta) \lambda_I
  - \delta \|\eta'\|_\infty^2 - \frac{\delta}{8(1+\beta^2)}
  \right]
  \|\psi\|_{\sii(\Omega_0^I)}^2
\end{align*}
for every real $\delta$.
Choosing (positive)~$\delta$ in such a way
that the square bracket vanishes,
we arrive at the global Hardy-type inequality
$$
  h[\psi] - E_1(\beta) \|\psi\|^2
  \geq 
  c' \int_{\Omega_0} 
  \frac{|\psi(s,t)|^2}{1+(s-s_0)^2} \, \der s \, \der t
$$
with
$$
  c' := 
  \frac{\lambda_I}{
  16(1+\beta^2) (\lambda_I + \|\eta'\|_\infty^2) + 2}
$$
From this inequality we also deduce 
\begin{equation}\label{Hardy.better}
  h[\psi] - E_1(\beta) \|\psi\|^2
  \geq 
  c \int_{\Omega_0} 
  \frac{|\psi(s,t)|^2}{1+s^2} \, \der s \, \der t
\end{equation}
with
$$
  c := c' \inf_{s\in\Real} \frac{1+s^2}{1+(s-s_0)^2}  
  \,,
$$
where the infimum is positive.
The desired inequality~\eqref{Hardy} follows
by the unitary equivalence between~$H$ and~$-\Delta_D^\Omega$
together with the fact that the the coordinates~$s$ and~$x$
are equivalent through this transformation, \cf~\eqref{tube}.
\end{proof}
\begin{Remark}
The present proof does not employ the presence 	of the first term
on the right-hand side of~\eqref{Hardy.local.improved}.
The inequality~\eqref{Hardy.better}
can be consequently improved to
$$
  H - E_1(\beta) \geq 
  \delta \, \beta\eps(s) 
  \left[
  E_1 + \left(\frac{\chi_1'(t)}{\chi_1(t)}\right)^2
  \right]
  + (1-\delta) \, \frac{c}{1+s^2}
$$
with any $\delta \in [0,1]$.
\end{Remark}
%

\section{Location of the essential spectrum}\label{Sec.ess} 
%
In this section we prove Theorem~\ref{Thm.ess}.	
Since the technical approaches for finite and infinite~$\beta$
are quite different, we accordingly split the section 
into two subsections.
In either case, we always assume $f' \in L_\mathrm{loc}^\infty(\Real)$. 

\subsection{Finite limits}
We employ the following characterisation of the essential spectrum,
for which we are inspired in~\cite[Lem.~4.2]{DDI}.

\begin{Lemma}\label{Lem.ess}
A real number~$\lambda$ belongs to the essential spectrum of~$H$
if, and only if, there exists a sequence 
$\{\psi_n\}_{n=1}^\infty \subset \Dom h$ such that
the following three conditions hold:
\begin{enumerate}
\item[\emph{(i)}]
$\|\psi_n\| = 1$ for every $n \geq 1$,
\item[\emph{(ii)}]
$
  (H-\lambda)\psi_n \to 0 
$
as $n\to\infty$
in the norm of the dual space~$(\Dom h)^*$.
\item[\emph{(iii)}]
$\supp \psi_n \subset \Omega_0 \setminus (-n,n)\times(0,d)$
for every $n \geq 1$.
\end{enumerate}
\end{Lemma}
\begin{proof}
By the general Weyl criterion modified to quadratic forms 
(\cf~\cite[Lem.~4.1]{DDI} or \cite[Thm.~5]{KL}
where the latter contains a proof),
$\lambda$ belongs to the essential spectrum of~$H$
if, and only if, there exists a sequence 
$\{\phi_n\}_{n=1}^\infty \subset \Dom h$ such that~(i) and~(ii) hold
but~(iii) is replaced by 
\begin{enumerate}
\item[(iii')]
$\phi_n \to 0$ as $n \to \infty$ weakly in $\sii(\Omega_0)$.
\end{enumerate}
The sequence $\{\psi_n\}_{n=1}^\infty$
satisfying~(i) and~(iii) is clearly weakly converging to zero. 
Hence, one implication of the lemma is obvious.
Conversely, let us assume that there exists 
a sequence $\{\phi_n\}_{n=1}^\infty \subset \Dom h$
satisfying~(i), (ii) and~(iii')
and let us construct from it a sequence $\{\psi_n\}_{n=1}^\infty \subset \Dom h$   
satisfying~(i), (ii) and~(iii).
Since $C_0^\infty(\Omega_0)$ is the form core of~$H$
(\cf~Proposition~\ref{Prop.core}),
we may assume that $\{\phi_n\}_{n\in\Nat} \subset C_0^\infty(\Omega_0)$.

Let us denote the norms of $\Dom h$ and $(\Dom h)^*$ by
$\|\cdot\|_{+1}$ and $\|\cdot\|_{-1}$, respectively.
One has $\|\cdot\|_{\pm 1} = \|(H+1)^{\pm 1/2}\cdot\|$.
By writing
\begin{equation}\label{DDI}
  \phi_n = (H+1)^{-1} (H-\lambda) \phi_n + (\lambda+1) (H+1)^{-1}\phi_n
\end{equation}
and using~(ii),
we see that the sequence $\{\phi_n\}_{n=1}^\infty$ 
is bounded in $\Dom h$.

Let $\eta \in C^\infty(\Real)$ be such that $0 \leq\eta\leq 1$,
$\eta=0$ on $[-1,1]$ and $\eta=1$ on $\Real\setminus(-2,2)$. 
For every $k \geq 1$, we set $\eta_k(s) := \eta(s/k)$
and we keep the same notation for the function 
$\eta_k \otimes 1$ on $\Real \times (0,d)$,
and similarly for its derivatives~$\eta_k'$ and~$\eta_k''$.
Clearly, $\supp \eta_k \subset \Omega_0 \setminus (-k,k)\times(0,d)$. 
For every $k \geq 1$, the operator $(1-\eta_k)(H+1)^{-1}$
is compact in $\sii(\Omega_0)$.
By virtue of the weak convergence~(iii'), 
it follows that, for every $k \geq 1$, 
$(1-\eta_k)(H+1)^{-1} \phi_n \to 0$ as $n \to \infty$
in $\sii(\Omega_0)$.
Then there exists a subsequence 
$\{\phi_{n_k}\}_{k=1}^\infty$ of $\{\phi_n\}_{n=1}^\infty$ 
such that $(1-\eta_k)(H+1)^{-1} \phi_{n_k} \to 0$ as $k \to \infty$
in $\sii(\Omega_0)$. 
Consequently,
the identity~\eqref{DDI} together with~(ii)
implies that $(1-\eta_k)\phi_{n_k} \to 0$ as $k \to \infty$ in $\sii(\Omega_0)$.
It follows that $\eta_k \phi_{n_k}$ can be normalised
for all sufficiently large~$k$. More specifically, 
redefining the subsequence if necessary, we may assume that
$\|\eta_k \phi_{n_k}\| \geq 1/2$ for all $k \geq 1$.
We set
$$
  \psi_k := \frac{\eta_k \phi_{n_k}}{\|\eta_k \phi_{n_k}\|}
$$
and observe that $\{\psi_k\}_{k=1}^\infty$ satisfies~(i) and~(iii). 

It remains to verify~(ii).
To this purpose, we notice that, using the duality, (ii)~means
\begin{equation}\label{dual.norm}
  \|(H-\lambda)\phi_n\|_{-1}
  = \sup_{\stackrel[\varphi\not=0]{}{\varphi\in\Dom h}} 
  \frac{|h(\varphi,\phi_n)-\lambda(\varphi,\phi_n)|}{\|\varphi\|_{+1}}
  \xrightarrow[n \to \infty]{} 0
  \,,
\end{equation}
where $\|\varphi\|_{+1}^2 = h[\varphi] + \|\varphi\|^2$.
By the direct computation employing integrations by parts,
one can check the identity 
\begin{multline*}
  h(\varphi,\eta_k \phi_{n_k}) - \lambda (\varphi,\eta_k \phi_{n_k})
  \\
  = h(\eta_k\varphi,\phi_{n_k}) - \lambda (\eta_k\varphi,\phi_{n_k})
  + 2(\partial_s\varphi-f'\partial_t\varphi,\eta_k'\phi_{n_k})
  + (\varphi,\eta_k''\phi_{n_k})
\end{multline*}
for every test function $\varphi \in C_0^\infty(\Omega_0)$,
a core of $\Dom h$.
Using~\eqref{dual.norm}, we have
$$
\begin{aligned}
  \sup_{\stackrel[\varphi\not=0]{}{\varphi\in C_0^\infty(\Omega_0)}} 
  \frac{|h(\eta_k\varphi,\phi_{n_k})
  -\lambda(\eta_k\varphi,\phi_{n_k})|}{\|\varphi\|_{+1}}
  &\leq \sup_{\stackrel[\eta_k\varphi\not=0]{}{\varphi\in C_0^\infty(\Omega_0)}} 
  \frac{|h(\eta_k\varphi,\phi_{n_k})
  -\lambda(\eta_k\varphi,\phi_{n_k})|}{\|\eta_k\varphi\|_{+1}}
  \\
  &= \|(H-\lambda)\phi_{n_k}\|_{-1}
  \xrightarrow[k \to \infty]{} 0
  \,.
\end{aligned}
$$
At the same time, using the Schwarz inequality
and the estimates 
$
  \|\partial_s\varphi-f'\partial_t\varphi\|^2 \leq h[\varphi] 
  \leq \|\varphi\|_{+1}^2
$,
we get
$$
\begin{aligned}
  \sup_{\stackrel[\varphi\not=0]{}{\varphi\in C_0^\infty(\Omega_0)}} 
  \frac{|2(\partial_s\varphi-f'\partial_t\varphi,\eta_k'\phi_{n_k})|}
  {\|\varphi\|_{+1}}
  &\leq \|\eta_k'\|_\infty \|\phi_{n_k}\|
  \,,
\end{aligned}
$$
where $\|\eta_k'\|_\infty$ denotes the supremum norm of~$\eta_k'$.
In view of the normalisation~(i) 
and since $\|\eta_k'\|_\infty = k^{-1} \|\eta'\|_\infty$,
we see that also this term tends to zero as $k \to \infty$.
Finally, using the Schwarz inequality 
and the estimate $\|\varphi\| \leq \|\varphi\|_{+1}$,
we have
\begin{equation*}
  \sup_{\stackrel[\varphi\not=0]{}{\varphi\in C_0^\infty(\Omega_0)}} 
  \frac{|(\varphi,\eta_k''\phi_{n_k})|}
  {\|\varphi\|_{+1}}
  \leq \|\eta_k''\|_\infty \|\phi_{n_k}\|
  = k^{-2} \|\eta''\|_\infty 
  \xrightarrow[k \to \infty]{} 0
  \,.
\end{equation*}
Summing up, we have just checked 
$\|(H-\lambda)(\eta_k \phi_{n_k})\|_{-1} \to 0$ as $k\to \infty$.
Recalling $\eta_k \phi_{n_k} \geq 1/2$,
the desired property~(ii) for $\{\psi_{k}\}_{k=1}^\infty$ follows. 
\end{proof}

Using this lemma, we immediately arrive at the following
``decomposition principle'' (saying that the essential spectrum
is determined by the behaviour at infinity only).

\begin{Proposition}\label{Prop.decomposition}
If~\eqref{Ass} holds with $\beta\in\Real$, then 
$
  \sigma_\mathrm{ess}(H_{\beta+\eps}) 
  = \sigma_\mathrm{ess}(H_{\beta}) 
$.
\end{Proposition}
\begin{proof}
Let $\lambda \in \sigma_\mathrm{ess}(H_{\beta})$.
By Lemma~\ref{Lem.ess}, there exists a sequence 
$\{\psi_n\}_{n=1}^\infty$ satisfying the properties (i)--(iii) 
with~$H$ being replaced by~$H_\beta$.
One easily checks the identity
\begin{equation}\label{id.decomposition}
  h_{\beta+\eps}(\varphi,\psi_n)
  = h_{\beta}(\varphi,\psi_n)
  - (\partial_s\varphi-\beta\partial_t\varphi,\eps\partial_t\psi_n)
  - (\eps\partial_t\varphi,\partial_s\psi_n-(\beta+\eps)\partial_t\psi_n)
\end{equation}
for every test function 
$
  \varphi \in \Dom h_{\beta}
  = H_0^1(\Omega_0)
  = \Dom h_{\beta+\eps}
$
(\cf~Proposition~\ref{Prop.core}).
Now we proceed similarly as in the end of the proof of Lemma~\ref{Lem.ess}.
Using the estimates 
$
  \|\partial_s\varphi-\beta\partial_t\varphi\|^2
  \leq h[\varphi]
  \leq \|\varphi\|_{+1}
$
and the Schwarz inequality,
we have
$$
\begin{aligned}
  \sup_{\stackrel[\varphi\not=0]{}{\varphi\in H_0^1(\Omega_0)}} 
  \frac{|(\partial_s\varphi-\beta\partial_t\varphi,\eps\partial_t\psi_{n})|}
  {\|\varphi\|_{+1}}
  &\leq \|\eps\partial_t\psi_{n}\|
  \leq \|\eps\|_{L^\infty(\Real\setminus(-n,n))} \|\partial_t\psi_{n}\|
  \xrightarrow[n \to \infty]{}
  0
  \,.
\end{aligned}
$$
Recall that $\{\psi_n\}_{n=1}^\infty$ is bounded in $H_0^1(\Omega_0)$,
\cf~\eqref{DDI} and Proposition~\ref{Prop.core}.
Similarly, using $\|\partial_t\varphi\| \leq \|\varphi\|_{+1}$
and the Schwarz inequality, we obtain
$$
\begin{aligned}
  \sup_{\stackrel[\varphi\not=0]{}{\varphi\in H_0^1(\Omega_0)}} 
  \frac{|(\eps\partial_t\varphi,\partial_s\psi_n-(\beta+\eps)\partial_t\psi_n)|}
  {\|\varphi\|_{+1}}
  &\leq \|\eps\partial_s\psi_{n}\| 
  + \|\beta+\eps\|_\infty \|\eps\partial_t\psi_{n}\|
  \\
  &\leq \|\eps\|_{L^\infty(\Real\setminus(-n,n))} 
  (\|\partial_s\psi_{n}\| 
  + \|\beta+\eps\|_\infty \|\partial_t\psi_{n}\|)
  \\
  & \xrightarrow[n \to \infty]{}
  0
  \,.
\end{aligned}
$$
Consequently, since $\|(H_{\beta}-\lambda)\psi_n\|_{-1} \to 0$
as $n\to\infty$ due to~(ii), 
it follows from~\eqref{id.decomposition} and~\eqref{dual.norm} that also
$\|(H_{\beta+\eps}-\lambda)\psi_n\|_{-1} \to 0$
as $n \to \infty$. 
Hence $\lambda \in \sigma_\mathrm{ess}(H_{\beta+\eps})$.

The opposite inclusion 
$
  \sigma_\mathrm{ess}(H_{\beta+\eps}) 
  \subset \sigma_\mathrm{ess}(H_{\beta}) 
$
is proved analogously.
\end{proof}

It remains to determine the (essential) spectrum 
for the constant shear.

\begin{Proposition}\label{Prop.Fourier}
One has
$
  \sigma(H_{\beta}) 
  = \sigma_\mathrm{ess}(H_{\beta}) 
  = [E_1(\beta),\infty)
$.
\end{Proposition}
\begin{proof}
By performing the partial Fourier transform 
in the longitudinal variable $s \in \Real$,
one has the unitary equivalence
$$
  H_\beta = -(\partial_s-\beta\partial_t)^2 - \partial_t^2 
  \cong -(i\xi-\beta\partial_t)^2 - \partial_t^2 
  \,,
$$
where $\xi \in \Real$ is the dual variable in the Fourier image.
Consequently, 
\begin{equation}\label{Fourier}
  \sigma(H_\beta) = \bigcup_{\xi \in \Real} \sigma(T_\beta(\xi))
  \,,
\end{equation}
where, for each fixed $\xi\in\Real$, 
$T_\beta(\xi)$ is the operator in $\sii((0,d))$
that acts as $-(i\xi-\beta\partial_t)^2 - \partial_t^2$
and satisfies Dirichlet boundary conditions.
The spectral problem for $T_\beta(\xi)$ can be solved explicitly. 
Alternatively,
one can proceed by ``completing the square'' 
and ``gauging out a constant magnetic field'' as follows:
$$
\begin{aligned}
  -(i\xi-\beta\partial_t)^2 - \partial_t^2
  &= (1+\beta^2) (- \partial_t^2)
  + 2i\xi\beta \partial_t
  + \xi^2
  \\
  &= - (1+\beta^2) \left(\partial_t-i \frac{\xi\beta}{1+\beta^2}\right)^2
  - \frac{\xi^2\beta^2}{1+\beta^2} + \xi^2	
  \\
  &= - (1+\beta^2) \left(\partial_t-i \frac{\xi\beta}{1+\beta^2}\right)^2
  + \frac{\xi^2}{1+\beta^2}
  \\
  &= 
  e^{iFt}
  \left[ - (1+\beta^2) \partial_t^2 + \frac{\xi^2}{1+\beta^2} \right]
  e^{-iFt}
  \,,
\end{aligned}
$$
where 
$$
  F := \frac{\xi\beta}{1+\beta^2}
  \,.
$$
Consequently, 
\begin{equation}\label{square} 
  \sigma(T_\beta(\xi)) = \left\{
  (1+\beta^2) E_n + \frac{\xi^2}{1+\beta^2}
  \right\}_{n=1}^\infty
  \,,
\end{equation}
where $E_n = n^2 E_1$ denotes the $n^\mathrm{th}$ eigenvalue
of the Dirichlet Laplacian in $\sii((0,d))$.
As a consequence of~\eqref{Fourier} and~\eqref{square},
we get the desired claim. 
\end{proof}

For finite~$\beta$,
Theorem~\ref{Thm.ess} follows as a direct consequence
of Propositions~\ref{Prop.decomposition} and~\ref{Prop.Fourier}.

\subsection{Infinite limits}
Let us now prove Theorem~\ref{Thm.ess} in the case $\beta = \pm\infty$.
Here the proof is based on the following purely geometric fact.

\begin{Proposition}\label{Prop.volume}
Let $f' \in L_\mathrm{loc}^\infty(\Real)$ satisfy~\eqref{Ass}
with $\beta = \pm\infty$. Then
\begin{equation}\label{volume} 
  \limsup_{\stackrel[\mathrm{x}\in\Omega]{}{|\mathrm{x}|\to\infty}} 
  |B_1(\mathrm{x}) \cap \Omega| = 0
  \,,
\end{equation}
where $B_1(\mathrm{x})$ denotes the disk of radius~$1$ 
centred at $\mathrm{x}\in\Real^2$
and~$|\cdot|$ stands for the Lebesgue measure.   
\end{Proposition}
\begin{proof}
Recall the diffeomorphism $\mathscr{L}:\Omega_0\to\Omega$ given by~\eqref{tube}.
For every number $\eta \in \Real$,
let $\Sigma(\eta)$ denote the set of points 
intersecting the boundary~$\partial\Omega$
with the straight horizontal half-line $(0,\infty) \times \{\eta\}$.
Since $f'(s) \to \pm\infty$ as $|s|\to\infty$, 
it follows that there exists a positive constant~$s_0$
such that~$f$ is either strictly increasing or strictly decreasing on $(s_0,\infty)$.
Let us consider the former situation, the latter can be treated analogously. 
It follows that there exists a positive constant~$\eta_0$ 
such that, for all $\eta \geq \eta_0$,
the set~$\Sigma(\eta)$ consists of just two points
$\mathrm{x}_1(\eta)=\mathscr{L}(s_1(\eta),d)$ 
and $\mathrm{x}_2(\eta)=\mathscr{L}(s_2(\eta),0)$
with $1 < s_1(\eta) < s_2(\eta)$.
Let $\mathrm{x}=\mathscr{L}(s,t)$ be a point in~$\Omega$ 
such that~$s$ is positive  
and the vertical component satisfies the inequality $\eta := f(s)+t \geq 1 + \eta_0$
(notice that $|\mathrm{x}| \to \infty$ if, and only if, $|s| \to \infty$).
Then we have the crude estimate
$$
  |B_1(\mathrm{x}) \cap \Omega|
  \leq 2 \, |s_2(\eta+1) - s_1(\eta-1)|
  \,.
$$
By the mean value theorem,
$$
  2 = f(s_2(\eta+1)) - f(s_1(\eta-1))
  = f'(\xi_\eta) \, |s_2(\eta+1) - s_1(\eta-1)|
  \,,
$$
where $\xi_\eta \in (s_1(\eta-1),s_2(\eta+1))$.
Since $s_1(\eta-1) \to \infty$ as $s \to \infty$,
we also have $\xi_\eta \to \infty$,
and therefore $|B_1(\mathrm{x}) \cap \Omega| \to 0$ as $s \to \infty$.
The case $s \to -\infty$ can be treated analogously.
\end{proof}

By the Berger-Schechter criterion 
(\cf~\cite[Thm.~V.5.17 \& Rem.~V.5.18.4]{Edmunds-Evans}),
it follows from~\eqref{volume}
that the embedding $W_0^{1,2}(\Omega) \hookrightarrow L^2(\Omega)$ is compact.
Consequently, the Dirichlet Laplacian $-\Delta_D^\Omega$ 
has no essential spectrum if $\beta = \pm\infty$
and Theorem~\ref{Thm.ess} is proved.

\begin{Remark}\label{Rem.q-bounded}
A necessary condition for~\eqref{volume} is that 
\begin{equation*}
  \lim_{\stackrel[\mathrm{x}\in\Omega]{}{|\mathrm{x}|\to\infty}} 
  \dist(\mathrm{x},\partial\Omega) = 0
  \,,
\end{equation*}
which means that~$\Omega$ is a \emph{quasi-bounded} domain 
(\cite[Sec.~X.6.1]{Edmunds-Evans}) in the regime $\beta = \pm\infty$.
\end{Remark}
%

\section{Existence of discrete spectrum}\label{Sec.disc}
%
In view of Theorem~\ref{Thm.ess},
the conclusion~\eqref{less} of Theorem~\ref{Thm.disc} guarantees
that any of the conditions~(i) or~(ii) implies 
that $-\Delta_D^\Omega$ possesses
at least one isolated eigenvalue of finite multiplicity below~$E_1(\beta)$. 
Let us establish these sufficient conditions.

\begin{proof}[Proof of Theorem~\ref{Thm.disc}]
Given any positive~$n$,
let $\varphi_n:\Real\to[0,1]$ be the continuous function
such that $\varphi_n=1$ on $[-n,n]$,
$\varphi_n=0$ on $\Real \setminus (-2n,2n)$
and linear in the remaining intervals.
Set 
$$
  \psi_n(s,t) := \varphi_n(s) \chi_1(t)
$$
and observe that $\psi_n \in \Dom(h) \cap H_0^1(\Omega_0)$.
Moreover, the identity of Lemma~\ref{Lem.fine} 
remains valid for~$\psi_n$ by approximation. 
Indeed, it is only important to notice that
$\chi_1^{-1}\psi_n \in \sii(\Omega_0)$
as a consequence of \cite[Thm.~1.5.6]{Davies_1989}.
Consequently,
\begin{equation*}
\begin{aligned}
  h[\psi_n] - E_1(\beta) \|\psi_n\|^2
  &=  \|\partial_s\psi_n
  - \eps \partial_t \psi_n\|^2 
  + \int_{\Omega_0} \beta\,\eps(s)\left[
  E_1 + \left(\frac{\chi_1'(t)}{\chi_1(t)}\right)^2
  \right]
  |\psi_n(s,t)|^2 \, \der s \, \der t
  \\
  &= \|\varphi_n'\|_{\sii(\Real)}^2
  + E_1 \int_{\Real} \left[\eps(s)^2 + 2 \beta \eps(s)\right]
  |\varphi_n(s)|^2 \, \der s
  \,,
\end{aligned}
\end{equation*}
where the second equality follows by the special form of~$\psi_n$
(the cross term vanishes due to an integration by parts with respect to~$t$)
and the formula $\|\chi_1'\|_{\sii((0,d))}^2=E_1$
together with the normalisation of~$\chi_1$.
Noticing that $\|\varphi_n'\|_{\sii(\Real)}^2=2n^{-1}$
and using the dominated convergence theorem,
we get
$$
  h[\psi_n] - E_1(\beta) \|\psi_n\|^2
  \xrightarrow[n\to\infty]{}
  E_1 \int_{\Real} \left[\eps(s)^2 + 2 \beta \eps(s)\right]
  \der s
  \,,
$$
where the right-hand side is negative by the hypothesis~(i).
Consequently, there exists a positive~$n$ such that
$h[\psi_n] - E_1(\beta) \|\psi_n\|^2 < 0$,
so the desired result~\eqref{less} follows by 
the variational characterisation of the lowest point
in the spectrum of~$H$. 

Now assume~(ii), in which case the shifted form
$
  h_1[\psi_n] := 
  h[\psi_n] - E_1(\beta) \|\psi_n\|^2
$ 
converges to zero as $n \to \infty$.
Then we modify the test function~$\psi_n$ 
by adding a small perturbation:
$$
  \psi_{n,\delta}(s,t) := \psi_n(s,t) + \delta \, \phi(s,t)
  \qquad \mbox{with} \qquad
  \phi(s,t) := \xi(s) \, t \, \chi_1(t)
  \,,
$$
where~$\delta$ is a real number 
and $\xi \in C_0^\infty(\Real)$ is a real-valued function
to be determined later.
Writing 
\begin{equation}\label{perturbation}
\begin{aligned}
  \lim_{n\to\infty} h_1[\psi_{n,\delta}] 
  &= \lim_{n\to\infty} \left(
  h_1[\psi_n] + 2\delta \, h_1(\psi_n,\phi) + \delta^2 \, h_1[\phi]
  \right)
  \\
  &= 2\delta \lim_{n\to\infty} h_1(\psi_n,\phi) + \delta^2 \, h_1[\phi]
  \,,
\end{aligned}
\end{equation}
it is enough to show, in order to establish~\eqref{less}, 
that the limit of $h_1(\psi_n,\phi)$ as $n\to\infty$ is non-zero
for a suitable choice of~$\xi$.
Indeed, it then suffices to choose~$\delta$ sufficiently small
and of suitable sign to make the second line of~\eqref{perturbation} negative,
and subsequently choose~$n$ sufficiently large to make 
the whole expression $h_1[\psi_{n,\delta}]$ negative. 
Employing Lemma~\ref{Lem.fine} 
and the fact that $\varphi_n=1$ on $\supp\xi$ 
for all sufficiently large~$n$, 
we have
\begin{equation*}
\begin{aligned}
  \lim_{n\to\infty} h_1(\psi_n,\phi) 
  = \ & \int_{\Omega_0}
  \Big\{
  \chi_1(t) \chi_1'(t) \, [\eps(s)^2\xi(s) + \beta \eps(s) \xi(s)]
  - \chi_1(t) \chi_1'(t) t \, \eps(s)\xi'(s)
  \\
  & + \chi'(t)^2 t \, [\eps(s)^2\xi(s)+\beta \eps(s) \xi(s)]
  + \chi_1(t)^2 t \, E_1 \beta \eps(s) \xi(s)
  \Big\}
  \, \der s \, \der t
  \\
  = \ &
  \frac{1}{2} \int_{\Real} 
  \left\{
  \eps(s)\xi'(s) + E_1 \, d \, [\eps(s)^2+2\beta\eps(s)] \, \xi(s)
  \right\}
  \der s
  \\
  = \ &
  \frac{1}{2} \int_{\Real} 
  \left\{
  -\eps'(s) + E_1 \, d \, [\eps(s)^2+2\beta\eps(s)] 
  \right\}
  \xi(s)
  \, \der s
  \,.
\end{aligned}
\end{equation*}
Here the last equality follows by an integration by parts
employing the extra hypothesis that the derivative~$\eps'$ 
exists as a locally integrable function. 
Now let assume, by contradiction, that the last integral
equals zero for all possible choices of~$\xi$.
Then~$\eps$ must solve the differential equation
$
  -\eps'(s) + E_1 \, d \, [\eps(s)^2+2\beta\eps(s)] = 0
$,
which admits the explicit one-parametric class of solutions
$$
  \eps_c(s) = \frac{2\beta}{c \, e^{-2 E_1 d \beta s}-1} 
  \,, \qquad 
  c \in \Real
  \,.
$$ 
The solutions~$\eps_c$ for non-zero~$c$ are not admissible
because $f'^2-\beta$ is either a positive (if $c>0$) 
or negative (if $c<0$) function, so its integral cannot be equal to zero 
(for positive~$c$ the function~$\eps_c$ additionally admits
a singularity at $s=-(2 E_1 d \beta)^{-1}\log c$),
while~$\eps_0$ is excluded by the requirement that~$f'$ is not constant. 
Hence, there exists a compactly supported~$\xi$ such that
the limit of $h_1(\psi_n,\phi)$ as $n\to\infty$ is non-zero
and the proof is concluded by the argument described below~\eqref{perturbation}. 
\end{proof}
%

\section{Strong shearing}\label{Sec.schema}
%
In this last section we prove Theorem~\ref{Thm.schema}.
Accordingly, let us assume that 
the shear admits the special form $f'(x) = \beta + \alpha \eps(x)$,
where $\alpha, \beta \in \Real$ and $\eps:\Real \to \Real$ 
is a function such that
$\supp \eps \subset [0,1]$ and $c_1 \leq \eps(x) \leq c_2$
for all $x\in[0,1]$ with some positive constants $c_1,c_2$.
We divide the proof into several subsections.

\subsection{Preliminaries}
If~$\beta$ is non-negative (respectively, non-positive),
then the desired stability result~\eqref{stability} 
obviously holds for all $\alpha \geq 0$ (respectively, $\alpha \leq 0$)
due to the repulsiveness of the shear (\cf~Theorem~\ref{Thm.Hardy}).  
On the other hand, the shear becomes attractive 
if~$\beta$ is positive (respectively, negative)
and~$\alpha$ is small in absolute value and negative (respectively, positive),
so~\eqref{stability} cannot hold in these cases (\cf~Theorem~\ref{Thm.disc}).     
The non-trivial part of Theorem~\ref{Thm.schema} therefore 
consists in the statement that~\eqref{stability} holds again
in the previous regimes provided that~$\alpha$ becomes large in absolute value.    

Without loss of generality, it is enough to prove the remaining claim 
for~$\beta$ positive, $\alpha$~negative and the shear function 
\begin{equation}\label{schema}
  f(x) := 
  \begin{cases}
    \beta x & \mbox{if} \quad x<0 \,,
    \\ 
    \beta x + \alpha \int_0^x \eps(s) \, \der s
    & \mbox{if} \quad 0 \leq x \leq 1 \,,
    \\
    \beta x + \alpha \int_0^1 \eps(s) \, \der s 
    & \mbox{if} \quad 1 < x \,.
  \end{cases}
\end{equation}
As above, we define~$\Omega$ by~\eqref{Omega}
and do not highlight the dependence on $\alpha$.
We use the notation $(x,y) \in \Real^2$ 
for the Cartesian coordinates in which we describe~$\Omega$.

\subsection{Subdomains decomposition}
Our strategy is to decompose~$\Omega$ into a union of three 
open subsets $\Omega_\mathrm{ext}$, $\Omega_\mathrm{ver}$, $\Omega_\mathrm{int}$
and of two connecting one-dimensional interfaces  
$\Sigma_\mathrm{ext}$, $\Sigma_\mathrm{int}$ 
on which we impose extra Neumann boundary conditions later on. 
The definition of the ($\alpha$-dependent) sets given below
is best followed by consulting Figure~\ref{figC}.

\begin{figure}[h]
\begin{center}
\includegraphics[width=0.9\textwidth]{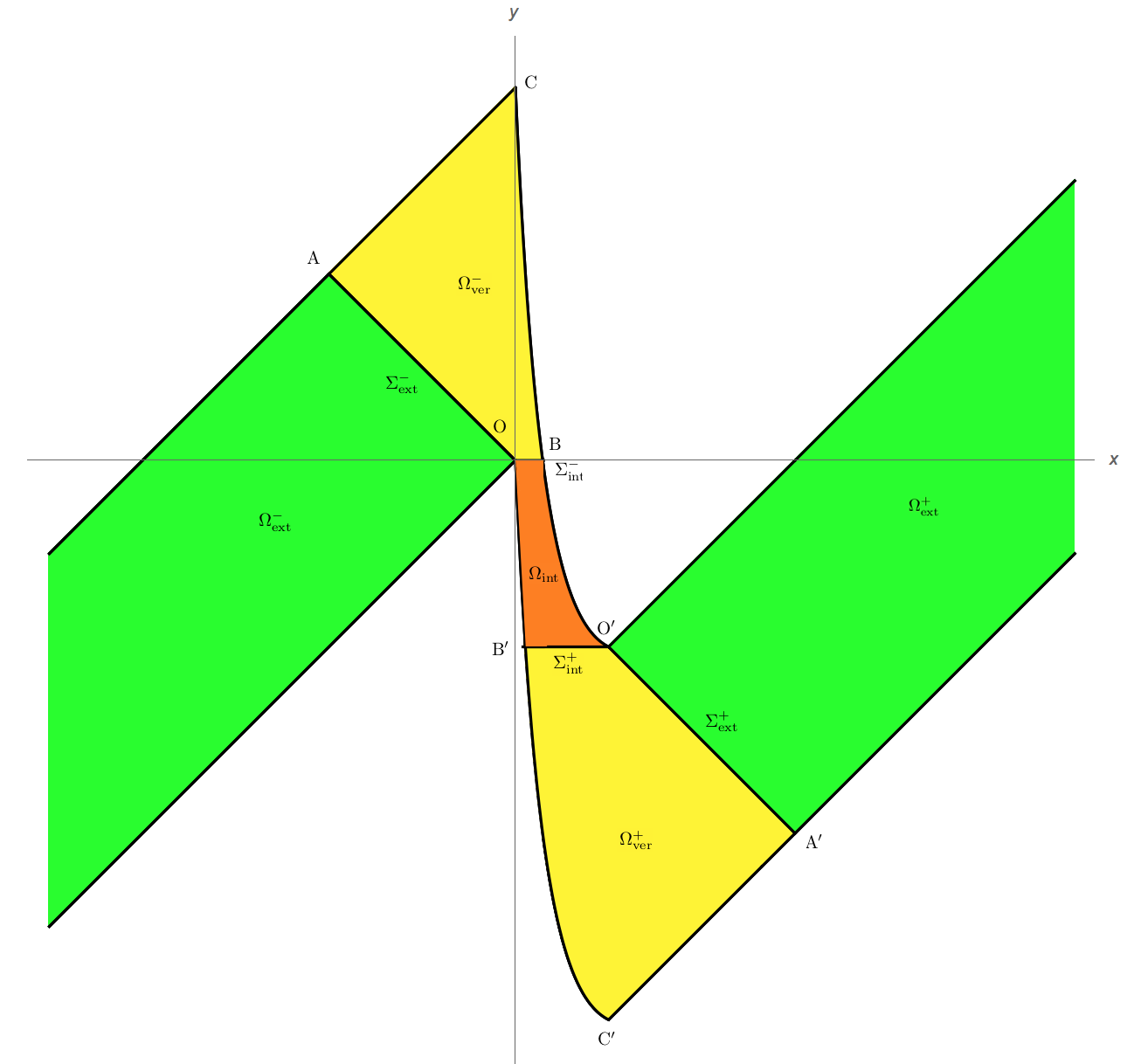}
\caption{The subdomains decomposition 
for~$\Omega$ corresponding to~\eqref{schema}.}\label{figC}
\end{center}
\end{figure}

Let $A\in\Real^2$ be the orthogonal projection of the origin $O:=(0,0) \in \Real^2$ 
on the half-line $\{(x,f(x)+d): x \leq 0\}$ 
and let $\Omega_\mathrm{ext}^-$ be the unbounded open subdomain of~$\Omega$
contained in the left half-plane and delimited by 
the open segment~$\Sigma_\mathrm{ext}^-$ connecting~$O$ with~$A$.
Similarly, let~$A'$ be the orthogonal projection of the point $O' := (1,f(1)+d)$
on the half-line $\{(x,f(x)): x \geq 1\}$ 
and let $\Omega_\mathrm{ext}^+$ be the unbounded open subdomain of~$\Omega$
contained in the right half-plane and delimited by 
the open segment~$\Sigma_\mathrm{ext}^{+}$ connecting~$O'$ with~$A'$.
Set $\Omega_\mathrm{ext} := \Omega_\mathrm{ext}^- \cup \Omega_\mathrm{ext}^+$
and $\Sigma_\mathrm{ext} := \Sigma_\mathrm{ext}^- \cup \Sigma_\mathrm{ext}^+$.

Let $x_0 \in (0,1)$ be the unique solution of the equation 
$f(x_0)+d = 0$ and set $B:=(x_0,0)$.
Similarly, let $x_0' \in (0,1)$ be the unique solution of the the equation
$f(x_0')=f(1)+d$ and set $B':=(x_0',f(1)+d)$.
Let us also introduce $C:=(0,f(0)+d)$ and $C':=(1,f(1))$.
Let~$\Omega_\mathrm{ver}^-$ (respectively, $\Omega_\mathrm{ver}^+$)
be the quadrilateral-like subdomain of~$\Omega$
determined by the vertices $O$, $A$, $C$ and $B$  
(respectively, $O'$, $A'$, $C'$ and $B'$).
Let~$\Sigma_\mathrm{int}^-$ (respectively, $\Sigma_\mathrm{int}^+$)
be the open segment connecting~$O$ with~$B$ 
(respectively, $O'$ with~$B'$).
Set $\Omega_\mathrm{ver} := \Omega_\mathrm{ver}^- \cup \Omega_\mathrm{ver}^+$
and $\Sigma_\mathrm{int} := \Sigma_\mathrm{int}^- \cup \Sigma_\mathrm{int}^+$.

Finally, we define 
$
  \Omega_\mathrm{int} := 
  \Omega\setminus(\Omega_\mathrm{ext}\cup\Omega_\mathrm{ver}
  \cup\Sigma_\mathrm{ext}\cup\Sigma_\mathrm{int})
$,
which is the parallelogram-like domain	 determined by the vertices $OBO'B'$.

\subsection{Neumann bracketing}
Notice that the majority of the sets introduced above depend on~$\alpha$,
although it is not explicitly highlighted by the notation.
In particular, the parallelogram-like domain determined by the vertices $OCO'C'$
(subset of $\Omega_\mathrm{ver} \cup \Omega_\mathrm{int}$)
converges in a sense to the half-line $\{0\} \times (-\infty,0)$
as $\alpha \to -\infty$.
Hence it is expected that this set is spectrally negligible in the limit
and the spectrum of $-\Delta_D^\Omega$ converges as $\alpha \to -\infty$
to the spectrum of the Dirichlet Laplacian in $\{(x,y)\in\Omega:x<0\}$.

Since we only need to show that the spectral threshold 
of $-\Delta_D^\Omega$ is bounded from below by~$E_1(\beta)$
for all sufficiently large~$\alpha$, 
in the following we prove less.
Instead, we impose extra Neumann boundary conditions on
$\Sigma_\mathrm{ext}\cup\Sigma_\mathrm{int}$
(\ie~no boundary condition in the form setting)
and show that the spectral threshold of the Laplacian
with combined Dirichlet and Neumann boundary conditions
in the decoupled subsets 
$\Omega_\mathrm{ext}$, $\Omega_\mathrm{ver}$, $\Omega_\mathrm{int}$
is bounded from below by~$E_1(\beta)$.
More specifically, 
by a standard Neumann bracketing argument
(\cf~\cite[Sec.~XIII.15]{RS4}),
we have 
$$
  \inf\sigma(-\Delta_D^\Omega) 
  \geq \min\left\{
  \inf\sigma(-\Delta_{DN}^{\Omega_\mathrm{ext}}),
  \inf\sigma(-\Delta_{DN}^{\Omega_\mathrm{ver}}),
  \inf\sigma(-\Delta_{DN}^{\Omega_\mathrm{int}})
  \right\}
  ,
$$ 
where $-\Delta_{DN}^{\Omega_\mathrm{ext}}$ 
(respectively, $-\Delta_{DN}^{\Omega_\mathrm{ver}}$;
respectively, $-\Delta_{DN}^{\Omega_\mathrm{int}}$)
is the operator in $\sii(\Omega_\mathrm{ext})$
(respectively, $\sii(\Omega_\mathrm{ver})$; 
respectively, $\sii(\Omega_\mathrm{int})$)
that acts as the Laplacian 
and satisfies Neumann boundary conditions on~$\Sigma_\mathrm{ext}$
(respectively, on $\Sigma_\mathrm{ext}\cup\Sigma_\mathrm{int}$;
respectively, on $\Sigma_\mathrm{int}$)
and Dirichlet boundary conditions 
on the remaining parts of the boundary.
It remains to study the spectral thresholds
of the individual operators. 

\subsection{Spectral threshold of the decomposed subsets}

\subsubsection{Exterior sets}
Notice that $\Omega_\mathrm{ext}$ consists of two connected components,
each of them being congruent 
to the half-strip $(0,\infty)\times(0,d/\sqrt{1+\beta^2})$.
Consequently, $-\Delta_{DN}^{\Omega_\mathrm{ext}}$ is isospectral to
the Laplacian in the half-strip, subject to Neumann boundary
conditions on $\{0\}\times(0,d/\sqrt{1+\beta^2})$
and Dirichlet boundary conditions on the remaining parts of the boundary.
By a separation of variables, it is straightforward to see that
$\sigma(-\Delta_{DN}^{\Omega_\mathrm{ext}})=[E_1(\beta),\infty)$.
In particular, the spectral threshold of $-\Delta_{DN}^{\Omega_\mathrm{ext}}$
equals $E_1(\beta)$ for all~$\alpha$.

\subsubsection{Interior set}
Given any $y \in [f(1)+d,0]$,
let $x_1(y)$ (respectively, $x_2(y)$) be the unique solution of the equation
$f(x_1(y))=y$ (respectively, $f(x_2(y))+d=y$).
Notice that $x_1(y) < x_2(y)$. 
In particular, $x_1(0)=0$ and $x_2(0)=x_0$,
while $x_1(f(1)+d)=x_0'$ and $x_2(f(1)+d)=f(1)+d$.
To estimate the spectral threshold of $-\Delta_{DN}^{\Omega_\mathrm{int}}$,
we write
$$
\begin{aligned}
  \int_{\Omega_\mathrm{int}} |\nabla\psi(x,y)|^2 \, \der x \, \der y
  &\geq \int_{\Omega_\mathrm{int}} |\partial_x\psi(x,y)|^2 \, \der x \, \der y
  \\
  &\geq \int_{\Omega_\mathrm{int}} 
  \left(\frac{\pi}{x_2(y)-x_1(y)}\right)^2
  |\psi(x,y)|^2 \, \der x \, \der y
  \,,
\end{aligned}
$$
where the second inequality follows by a Poincar\'e inequality
of the type~\eqref{Poincare} and Fubini's theorem.
Subtracting the equations that $x_1(y)$ and $x_2(y)$ satisfy,
the mean value theorem yields 
$d=f(x_1(y))-f(x_2(y)) = -f'(\xi) [x_2(y)-x_1(y)]$
with $0<\xi<x_0<1$.
Since $f'(x)=\beta+\alpha\eps(x) \geq \beta - c_2 |\alpha|$ for $x \in [0,1]$,
it follows that, for all $\alpha < \beta/c_2$, 
we have the uniform bound
$
  x_2(y)-x_1(y) \geq d / (c_2|\alpha|-\beta)
$.
Consequently, 
$\inf\sigma(-\Delta_{DN}^{\Omega_\mathrm{int}}) \to \infty$ 
as $\alpha \to -\infty$.
In particular, the spectral threshold of $-\Delta_{DN}^{\Omega_\mathrm{int}}$
is greater than $E_1(\beta)$ for all negative~$\alpha$ 
with sufficiently large~$|\alpha|$.

\subsubsection{Verge sets}
The decisive set~$\Omega_\mathrm{ver}$ consists of two connected 
components $\Omega_\mathrm{ver}^\pm$.
We consider $\Omega_\mathrm{ver}^-$, 
the argument for $\Omega_\mathrm{ver}^+$ being analogous.
Let us thus study the spectral threshold of 
the operator $-\Delta_{DN}^{\Omega_\mathrm{ver}^-}$ in $\sii(\Omega_\mathrm{ver}^-)$
that acts as the Laplacian and satisfies Neumann boundary conditions on
$\Sigma_\mathrm{ext}^- \cup \Sigma_\mathrm{int}^-$ 
and Dirichlet boundary conditions on the remaining parts 
of the boundary $\partial\Omega_\mathrm{ver}^-$.
Since the spectrum of this operator is purely discrete,
we are interested in the behaviour of its lowest eigenvalue.

As $\alpha \to -\infty$, the set $\Omega_\mathrm{ver}^-$
converges in a sense to the open right triangle~$T$ 
determined by the vertices $OAC$.
More specifically, $|\Omega_\mathrm{ver}^- \setminus T|=O(|\alpha|^{-1})$
as $\alpha \to -\infty$. 
Using the convergence result \cite[Thm.~29]{Barbatis-Burenkov-Lamberti_2010},
it particularly follows that the lowest eigenvalue of 
$-\Delta_{DN}^{\Omega_\mathrm{ver}^-}$
converges to the lowest eigenvalue $\lambda_1(T)$ 
of the Laplacian in~$T$,
subject to Neumann boundary conditions on the segment~$OA$
and Dirichlet boundary conditions on the other parts of the boundary. 
It remains to study the latter.

By a trivial-extension argument, 
$\lambda_1(T)$ is bounded from below by the lowest eigenvalue of 
the Laplacian in the rectangle $(0,|OA|) \times (0,|AC|)$,
subject to Neumann boundary conditions on $(0,|OA|) \times \{0\}$
and Dirichlet boundary conditions on the other parts of the boundary.
That is, 
$
  \lambda_1(T) \geq (1+\beta^{-2}) \, E_1(\beta)
$.

Summing up, we have established the result
$$
  \inf\sigma(-\Delta_{DN}^{\Omega_\mathrm{ver}})
  \xrightarrow[\alpha\to-\infty]{}
  \lambda_1(T) \geq (1+\beta^{-2}) \, E_1(\beta)
  \,.
$$
Consequently, there exists a negative~$\alpha_0$ such that,
for every $\alpha \leq \alpha_0$, the spectral threshold 
of $-\Delta_{DN}^{\Omega_\mathrm{ver}}$ is also bounded
from below by $E_1(\beta)$.
This concludes the proof of Theorem~\ref{Thm.schema}.


\subsection*{Acknowledgment}
%
D.K.\ was partially supported 
by the GACR grant No.\ 18-08835S
and by FCT (Portugal)
through project PTDC/MAT-CAL/4334/2014.

%
\bibliography{bib}
\bibliographystyle{amsplain}

\end{document}